\documentclass[suppldata]{interact}

\usepackage{epstopdf}
\usepackage[caption=false]{subfig}

\usepackage[numbers,sort&compress]{natbib}
\bibpunct[, ]{[}{]}{,}{n}{,}{,}
\renewcommand\bibfont{\fontsize{10}{12}\selectfont}
\makeatletter
\def\NAT@def@citea{\def\@citea{\NAT@separator}}
\makeatother

\theoremstyle{plain}
\newtheorem{theorem}{Theorem}[section]
\newtheorem{lemma}[theorem]{Lemma}
\newtheorem{corollary}[theorem]{Corollary}

\theoremstyle{definition}
\newtheorem{definition}[theorem]{Definition}

\theoremstyle{remark}
\newtheorem{remark}{Remark}

\usepackage{amssymb}
\usepackage{verbatim}
\usepackage{blkarray}
\usepackage{multirow}
\usepackage{graphicx}
\usepackage{afterpage}
\usepackage[hidelinks]{hyperref}
\usepackage{framed}
\usepackage{amsmath}
\usepackage{systeme}
\usepackage{xcolor}
\usepackage[all]{xy}
\usepackage{todonotes}

\newcommand{\ket}[1]{|#1\rangle}

\newcommand{\vp}{\varphi}

\newcommand{\img}{\operatorname{img}}

\newcommand{\f}[2]{\mathbb{F}_{#1}^{#2}}
\newcommand{\w}[1]{\mathbf{#1}}
\newcommand{\dhat}[1]{\widehat{\widehat{#1\,}}\!\!}

\DeclareMathOperator{\GPK}{GPK}

\begin{document}

\title{Two Problems on Quantum Computing in Finite Abelian Groups}

\author{
\name{Ulises Pastor--D\'iaz\textsuperscript{{$\dagger$}}\thanks{\textsuperscript{$\dagger$} CONTACT: Ulises Pastor--D\'iaz. Email: ulises.pastor@cunef.edu, ORCiD: 0000-0002-0309-7173} and Jos\'e~M. Tornero\textsuperscript{{$\ddagger$}}\thanks{\textsuperscript{$\ddagger$} Email: tornero@us.es, ORCiD: 0000-0001-5898-1049}}
\affil{\textsuperscript{$\dagger$} Departament of Mathematics. CUNEF University. C. de Leonardo Prieto Castro, 2, 28040 Madrid (Spain).}
\affil{\textsuperscript{$\ddagger$} Departamento de Álgebra, Facultad de Matemáticas, Universidad de Sevilla. Avda. Reina Mercedes s/n, 41012 Sevilla (Spain).}
}

\maketitle

\begin{abstract}
In the context of finite Abelian groups two problems are presented and solved using quantum computing techniques. The first is the well--known Hidden Subgroup Problem, originally solved by Simon in a landmark work. The second is the Fully Balanced Image Problem, originally introduced by the authors (joint with J. Ossorio--Castillo), which is related to a certain class of mappings (which contains strictly, for instance, the family of group morphisms). 

Both problems are tackled using a combination of two techniques: first, a conversion into Boolean objects, better suited for quantum computing arguments, and subsequently a custom--tailored algorithm which takes advantage of the Generalised Phase--Kick Back technique. 
\end{abstract}


\begin{amscode}
68Q12 (primary); 68Q09, 81P68 (secondary).
\end{amscode}

\begin{section}{Introduction}\label{uno}

\subsection{Notation}

Elements in $\f{2}{n}$ will be notated in bold to highlight their nature as vectors. The vector of zeros will be denoted as $\w{0}\in\f{2}{n}$, $\oplus$ will be used for the sum of vectors in $\f{2}{n}$ and $\cdot$ will denote the inner product in $\f{2}{n}$.

\ 

Given a set $S$, its indicator function will be noted as
$$
1_S(x) = \begin{cases}
    1 & \text{ if } x\in S\\
    0 & \text{ otherwise}.
\end{cases}
$$

For qubits we will use the standard bra--ket notation (as in \cite{niel,via,kaye}, for example).

\ 

We will use the notion of multisets, which are nothing but sets plus multiplicity. More formally, given a set $X$, a multiset on $X$ is a pair $(X,m)$, where 
$$
m: X \longrightarrow \mathbb{Z}_{\geq 0}
$$
is a function called multiplicity of the multiset. The basic idea is that of a set where we allow repetition of elements (including ignoring an element eventually if its multiplicity is $0$). The subset of elements with positive multiplicity is called the support of the multiset, noted $\sup(X,m)$. 

An extensive overview on multisets and their properties can be found in \cite{syro}. 

\subsection{A word on finite Abelian groups and quantum computing}

The main object of study will be finite Abelian groups (usually noted $G$ and $H$ with the operation $+$), whose elements we will need to represent as quantum states. For this purpose, given $G$, we will fix a string representation function, that is, an injective function 
$$
\vp_G : G\to\f{2}{n}
$$ 
such that the addition in $G$ can be constructed as an operator in the image $\img( \vp_G )$ and such that the identity element gets sent to $\w{0}\in\f{2}{n}$. From now on, we will abuse notation and write $G \subset \f{2}{n}$ instead of $\img(\vp_G)$ or $\vp_G(G)$. Likewise, $\ket{\w{g}}_G$, where $g\in G$, will represent the state $\ket{\vp_G(g)}_n$, and the operator in $\vp(G)$, which corresponds to addition in $G$, will be simply noted as $+$ or $+_G$ when clarification is needed. Although this operator only makes sense when applied to vectors in $\vp_G(G)\subset\f{2}{n}$, we can think of it as an operator defined for any vector in $\f{2}{n}$ that, in particular, codifies the addition of elements of $G$ in $\vp_G(G)$.

It is well known \cite{via} that for any Boolean function, $f:\f{2}{n}\to\f{2}{m}$, one can construct the quantum gate $\w{U}_f$ whose effect is:
$$
\w{U}_f \Big( \ket{\w{x}}_n \otimes \ket{\w{y}}_m \Big) = \ket{\w{x}}_n \otimes \ket{\w{y} \oplus f(\w{x})}_m.
$$

In the same spirit, this construction can be extended to generic functions $f: G \to H$ between arbitrary finite Abelian groups, via the previously described string representations.

The way this function is constructed is by first considering a function $\w{f}:\f{2}{n}\to\f{2}{m}$ that represents $f$, that is, one that verifies
$$
\w{f} (\w{g}) = \w{f}(\vp_G(g)) = \vp_H(f(g)), \quad \forall g \in G.
$$
$$
\xymatrix{
G \ar[r]^-{f}\ar[d]_-{\vp_G} & H \ar[d]^-{\vp_H} \\
\f{2}{n} \ar[r]^-{\w{f}} & \f{2}{m} } 
$$

Fixing one such $\w{f}$, we can now define $\w{U}_f$ as the quantum gate whose effect is:
$$
\w{U}_f\Big(\ket{\w{x}}_n\otimes\ket{\w{y}}_m\Big) = \ket{\w{x}}_n \otimes \ket{\w{y} +_H \w{f}(\w{x})}_m
$$
for any $\w{x} \in \f{2}{n}$ and any $\w{y} \in \f{2}{m}$. In particular, if $g\in G$ and $h\in H$, then:
$$
\begin{array}{rcl}
\w{U}_f\Big(\ket{\w{g}}_n\otimes\ket{\w{h}}_m\Big) &=& \ket{\w{g}}_n \otimes \ket{\vp_H(h)+_H\vp_H\big(f(g)\big)}_m \\[5pt]
&=& \ket{\w{g}}_n \otimes \ket{\vp_H\big(h+ f(g)\big)}_m = \ket{\w{g}}_n \otimes \ket{\w{h} + \w{f}(\w{g})}_m.
\end{array}
$$

To obtain this gate, we begin by taking two extra registers of $H$ and applying
$$
\w{U}_f\Big(\ket{\w{g}}_G\otimes\ket{\w{0}}_H\Big) = \ket{\w{g}}_G\otimes\ket{\w{f}(\w{g}) }_H
$$
if $g\in G$, and then we use the reversible gate corresponding to the operation in $\vp(H)$, $\w{U}_{+_H}$, yielding
$$
\w{U}_{+_H}\Big(\ket{\w{h}}_H\otimes\ket{\w{f(\w{g})}}_H\otimes\ket{\w{0}}_H\Big) = \ket{\w{h}}_H\otimes\ket{\w{f(\w{g})}}_H\otimes\ket{\w{h}+_H\w{f(\w{g})}}_H.
$$

Finally, we just need to discard the unnecessary registers. From this point onward, any $+$ is assumed to be $+_H$ unless stated otherwise.

\subsection{Characters}

Given a certain finite Abelian group $G$, a character of $G$ is just a group morphism $\chi:G \to S^1$. For characters of finite Abelian groups, references such as \cite{char} can be consulted, but a quick overview on the necessary results for this paper will be provided here. 

The characters of $G$ form a (multiplicative) group, denoted as $\widehat{G}$, with unit element $1_G$, called the trivial character, which is the constant function $1_G (g)= 1$, for all $g \in G$. The most interesting fact for our concerns is that $G$ and $\widehat{G}$ are isomorphic groups. This can be easily proved using the fact that any finite Abelian group is isomorphic to a product of cyclic groups. As this isomorphism is \textit{not} canonical, in order to simplify notation, we will assume that one such isomorphism has been considered and we will refer to the character identified with $g\in G$ as $\chi_g \in \widehat{G}$. In particular, $\chi_0 = 1_G$.

Given a character $\chi_g\in \widehat{G}$, its inverse element in the group is the character $\overline{\chi}_g$, which is the one defined as 
$$
\overline{\chi}_g(h) = (\chi_g(h))^{-1} = \chi_{g} \left( -h \right), \qquad \forall h\in G,
$$
the notation resembling the complex conjugation, as, being $\chi_g(h)$ unimodular,
$$
\chi_g(h)^{-1} = \overline{\chi_g(h)}.
$$

Given a subgroup $H\subset G$, one can consider 
$$
H^{\perp} = \left\{ \chi\in \widehat{G} \; \mid \; \chi(g) = 1, \text{ for all } g\in H \right\},
$$ 
which is a subgroup of $\widehat{G}$ with cardinality $|G|/|H|$. It should also be noted that, given a character $\chi\in\widehat{G}$ and a subgroup $H\subset G$, 
$$
\sum_{h\in H} \chi(h) = \begin{cases}
    |H| & \text{ if } \chi_{|H} = 1_H,\\
    0 & \text{ if } \chi_{|H} \neq 1_H.
\end{cases}
$$

Given then a function $f:G\to\mathbb{C}$, we can use characters to define its Fourier transform $\widehat{f}:\widehat{G}\to\mathbb{C}$ as
$$
\widehat{f}(\chi) = \sum_{g\in G} f(g)\overline{\chi}(g).
$$

An important result that will be used in Section \ref{cuatro} is that of the inverse Fourier transform: given a function $f:G\to\mathbb{C}$, we have
$$
\dhat{f} = |G|f.
$$

Also, when taking the Fourier transform of the indicator function of a subgroup, $H\subset G$, the result is:
$$
\widehat{1_{H}} (\chi) = \sum_{g\in G} 1_{H}(g)\chi(g) = \sum_{g\in H} \chi(g) = \begin{cases}
    |H| & \text{ if } \chi_{|H} = 1_{|H},\\
    0 & \text{ otherwise, }
\end{cases}
$$
that is, $\widehat{1_{H}} = |H|1_{H^{\perp}}$.

We will also assume that, for any Abelian group $G$, we can construct its Fourier transform as a quantum gate, $\w{F}_G$. The effect of such a gate is as follows:
$$
\w{F}_G\ket{\w{g}}_G = \frac{1}{\sqrt{|G|}}\sum_{z\in G}\chi_g(z)\ket{\w{z}}_G,
$$
and we will denote this state by $\ket{\chi_g}_G$.

In particular, if $G = \f{2}{n}$, then $\w{F}_G$ is the Hadamard gate. The behavior of the Fourier transform gate is thus analogous to that of the Haddamard gate: it transform an elements of the computational basis in superpositions of the elements of the computational basis with the same probability but (possibly) different and periodic amplitudes. In particular,
$$
\w{F}_G\ket{\w{0}}_G = \frac{1}{\sqrt{|G|}}\sum_{g\in G}\ket{\w{g}}_G.
$$

This gate was first introduced in \cite{abelian} and was later famously used by Peter Shor to solve the factorisation problem for integer numbers in \cite{shorpre}. A formal construction of this gate can be found in \cite{niel,via}, among others.

\subsection{GPK Algorithm}

The Generalised Phase-Kick Back or $\GPK$ algorithm for binary vector spaces $\f{2}{n}$, was introduced in \cite{gpk} and further studied in \cite{gpk2}. 
The following paragraphs will be devoted to providing a general overview on the $\GPK$ algorithm and its main applications. 

\ 

The algorithm itself takes as instances an arbitrary Boolean function, $f:\f{2}{n}\to\f{2}{m}$, and a binary string $\w{y}\in\f{2}{m}$, which will be known as the marker. It works as follows:

\vspace{5mm}

$\mathbb{STEP}$ $1$:
$$
\ket{\varphi_0}_{n,m} = \ket{\mathbf{0}}_n\otimes\ket{\mathbf{y}}_m.
$$

\vspace{3mm}

$\mathbb{STEP}$ $2$:
$$
\ket{\varphi_1}_{n,m} = \mathbf{H}_{n+m}\ket{\varphi_0}_{n,m} = \left(\displaystyle\frac{1}{\sqrt{N}}\displaystyle\sum_{\mathbf{x}\in\f{2}{n}} \ket{\mathbf{x}}_n\right)\otimes\ket{\gamma_{\mathbf{y}}}_m,
$$
where $\ket{\gamma_{\w{y}}}_m = \w{H}_n\ket{\w{y}}_m$ and $N = 2^n$.

\vspace{3mm}

$\mathbb{STEP}$ $3$:
$$
\ket{\varphi_2}_{n,m} = \mathbf{U}_f\ket{\varphi_1}_{n,m} = \left(\displaystyle\frac{1}{\sqrt{N}}\displaystyle\sum_{\mathbf{x}\in\f{2}{n}} (-1)^{\mathbf{y}\cdot f(\mathbf{x})}\ket{\mathbf{x}}_n\right)\otimes\ket{\gamma_{\mathbf{y}}}_m.
$$

\vspace{3mm}

$\mathbb{STEP}$ $4$:
$$
\ket{\varphi_3}_{n,m} = \left(\mathbf{H}^{\otimes n}\otimes \mathbf{I}^{\otimes m}\right)\ket{\varphi_2}_{n,m}.
$$
That is, 
$$
\ket{\varphi_3}_{n,m} = \displaystyle\frac{1}{N}\displaystyle\sum_{\mathbf{z}\in\f{2}{n}} \left(\displaystyle\sum_{\mathbf{x}\in\f{2}{n}} (-1)^{\mathbf{y}\cdot f(\mathbf{x})\oplus\mathbf{x}\cdot\mathbf{z}}\right)\ket{\mathbf{z}}_n\otimes\ket{\gamma_{\mathbf{y}}}_m.
$$

\vspace{3mm}

At this point, the second register can be discarded.

\vspace{3mm}

$\mathbb{STEP}$ $5$:
We measure the first register and name the result $\delta$.

\vspace{5mm}

One important application of the $\GPK$ is in solving the Hidden Subgroup Problem or HSP for vectorial Boolean functions, and, more particularly, Simon's problem \cite{simon}. The HSP will be further analysed in Section \ref{tres}, where its general version for Abelian groups will be tackled, but the solution and analysis to the Boolean function version of the problem can be found in \cite{gpk}. It should also be noted that, although the $\GPK$ is proved to be more efficient in solving this problem than Simon's algorithm, this improvement does not so far have any implications on the quantum complexity of the HSP or Simon's problem.

\ 

The other relevant application of the $\GPK$ is in solving what was defined in \cite{gpk2} as the Fully Balanced Image or FBI problem. The class of FBI functions was defined and analysed in \cite{waifi}, where it was shown that an FBI function is essentially a function $f:\f{2}{n}\to\f{2}{m}$ whose image is an affine space in $\f{2}{m}$ for which all pre-images have the same cardinality. The FBI problem is thus defined as the problem of, given one such function as a black box, determining the dimension of the given function image. 

One algorithm to solve this problem that employed the $\GPK$ as a subroutine was proposed in \cite{gpk2}, where it was shown that it was optimal among algorithms that follow a similar strategy. The efficiency of the algorithm was also analysed, and it was proven that it solved the problem in $\mathcal{O}\left(2^r(m-r+1)-1\right)$ applications of the $\GPK$, where $r$ is the desired dimension. The situation for Abelian groups will be treated in Section \ref{cuatro}.

\end{section}

\begin{section}{The Generalised Phase Kick-Back for Finite Abelian Groups}\label{dos}

This section will be devoted to defining and extending the idea of the Generalised Phase Kick-Back to finite Abelian groups. We will begin by studying the behavior of the Phase Kick-Back itself, so we can then apply it to the actual algorithm. Finally, we will show that the algorithm actually works as claimed.

\begin{lemma}\label{pk}
Let $G$ and $H$ be finite Abelian groups, and let $f:G\to H$ be a general function whose quantum gate is $\w{U}_f$. Then, any state of form $\ket{\w{g}}_G\otimes\ket{\w{\chi}_h}_H$, where $g\in G$ and $h\in H$, is an eigenvector of $\w{U}_f$ with eigenvalue $\overline{\chi_h}\left(f\left(g\right)\right)$.
\end{lemma}

\begin{proof}
We have that
$$
\ket{\w{g}}_G\otimes\ket{\w{\chi}_h}_H = \ket{\w{g}}_G\otimes\left(\frac{1}{\sqrt{|H|}}\sum_{x\in H}\chi_h(x)\ket{\w{x}}_H\right),
$$
and so if we apply $\w{U}_f$, we get
$$
\w{U}_f\left(\ket{\w{g}}_G\otimes\ket{\w{\chi}_h}_H\right) = \ket{\w{g}}_G\otimes\left(\frac{1}{\sqrt{|H|}}\sum_{x\in H}\chi_h(x)\ket{\w{x}+f(\w{g})}_H\right).
$$
Renaming now $x = y - f(g)$, we get the state:
$$
\ket{\w{g}}_G\otimes\left(\frac{1}{\sqrt{|H|}}\sum_{y\in H}\chi_h\big(y-f(g)\big)\ket{\w{y}}_H\right)
$$
$$
= \ket{\w{g}}_G\otimes\left(\overline{\chi_h}\big(f(g)\big)\frac{1}{\sqrt{|H|}}\sum_{y\in H}\chi_h(y)\ket{\w{y}}_H\right) = \overline{\chi_h}\big(f(g)\big)\Big(\ket{\w{g}}_G\otimes\ket{\w{\chi}_h}_H\Big).
$$
\end{proof}

Let us now move to the version of the algorithm to finite Abelian groups.

\begin{definition}{(Generalised Phase Kick-Back algorithm)}
Let $G, H$ be Abelian groups, and let  $f:G\to H$ be a function. The $\GPK$ algorithm applied to the function $f$ with marker $h\in H$---or rather $\chi_h\in \widehat{H}$---is the following one.

\vspace{5mm}

$\mathbb{STEP}$ $1$

$\ket{\varphi_0}_{G,H} = \ket{\mathbf{0}}_G\otimes\ket{\mathbf{h}}_H.$

\vspace{5mm}

$\mathbb{STEP}$ $2$

$\ket{\varphi_1}_{G,H} = \Big(\mathbf{F}_{G}\otimes\w{F}_{H}\Big)\ket{\varphi_0}_{G,H} = \left(\displaystyle\frac{1}{\sqrt{|G|}}\displaystyle\sum_{g\in G} \ket{\mathbf{g}}_G\right)\otimes\ket{\w{\chi}_h}_H.$

\vspace{5mm}

$\mathbb{STEP}$ $3$

$\ket{\varphi_2}_{G,H} = \mathbf{U}_f\ket{\varphi_1}_{G,H} = \left(\displaystyle\frac{1}{\sqrt{|G|}}\displaystyle\sum_{g\in G} \overline{\chi_h}\big(f(g)\big)\ket{\mathbf{g}}_G\right)\otimes\ket{\w{\chi}_h}_H.$

\vspace{5mm}

$\mathbb{STEP}$ $4$

$\ket{\varphi_3}_{G,H} = \left(\mathbf{F}_G\otimes \mathbf{I}^{\otimes m}\right)\ket{\varphi_2}_{G,H} = \displaystyle\frac{1}{|G|}\displaystyle\sum_{z\in G} \left(\displaystyle\sum_{g\in G} \overline{\chi_h}\big(f(g)\big)\chi_g(z)\right)\ket{\mathbf{z}}_G\otimes\ket{\w{\chi}_h}_H.$

\vspace{5mm}

At this point, the second register can be discarded.

\vspace{5mm}

$\mathbb{STEP}$ $5$

We measure the first register and name the result $\delta$.

\end{definition}

We have made some claims in the previous definition regarding the final states after performing the corresponding operations. In particular, we have stated that the final amplitude of the quantum state associated with an element of the Abelian group $z\in G$ is:
$$
\alpha_z = \frac{1}{|G|}\displaystyle\sum_{g\in G} \overline{\chi_h}\big(f(g)\big)\chi_g(z).
$$

Let us show that this is actually the case.

\begin{theorem}{(Correctness of the $\GPK$ algorithm.)}
    The amplitude of the quantum state associated with an element of the Abelian group $z\in G$ after applying the $\GPK$ under the previously stated circumstances is
    $$
    \alpha_z = \frac{1}{|G|}\displaystyle\sum_{g\in G} \overline{\chi_h}\big(f(g)\big)\chi_g(z).
    $$
\end{theorem}

\begin{proof}
    The final superpositions after steps $2$ and $3$ are trivial thanks to Lemma \ref{pk} and the definition of $\w{F}_G$. To prove the statement, we must thus focus on the result of the operation:
    $$
    \w{F}_G\frac{1}{\sqrt{|G|}}\displaystyle\sum_{g\in G} \overline{\chi_h}\big(f(g)\big)\ket{\mathbf{g}}_G = \frac{1}{\sqrt{|G|}}\displaystyle\sum_{g\in G} \overline{\chi_h}\big(f(g)\big)\w{F}_G\ket{\mathbf{g}}_G
    $$
    $$
    = \frac{1}{|G|}\displaystyle\sum_{g\in G} \overline{\chi_h}\big(f(g)\big)\Big(\sum_{z\in G}\chi_g(z)\ket{\w{z}}_G \Big)
    $$
    and by simple rearrangement we get:
    $$
    \frac{1}{|G|}\displaystyle\sum_{z\in G} \left(\displaystyle\sum_{g\in G} \overline{\chi_h}\big(f(g)\big)\chi_g(z)\right)\ket{\mathbf{z}}_G
    $$
\end{proof}

In particular, in Section \ref{cuatro} we will need to analyse the amplitude of the state $\ket{\w{0}}_G$, which is:
$$
\alpha_0 = \frac{1}{|G|}\sum_{g\in G} \overline{\chi_h}\big(f(g)\big).
$$
\end{section}

\begin{section}{The Hidden Subgroup Problem for finite Abelian Groups}\label{tres}

Let us now move to study the application of the $\GPK$ algorithm to a very relevant problem: the Hidden Subgroup Problem (HSP). For information on this problem and the current proposed solutions to it, \cite{hsp,hsp2,calixto} can be consulted.

This section will be structured as follows. Firstly, an overview of the HSP and of Simon's algorithm will be provided. Simon's algorithm for Abelian groups is a general version of Shor's algorithm, and it is an efficient algorithm that solves the HSP for the Abelian situation.

Secondly, we will show how the $\GPK$ behaves when targeting a function $f:G\to H$ that satisfies the HSP condition. Thirdly, it will be shown that the $\GPK$ algorithm is slightly better in a strict sense than Simon's algorithm in solving this problem. Unfortunately, this improvement is a small one, and does not have implications on the quantum complexity of the HSP. 

Finally, a few remarks on the possible future lines of research regarding the $\GPK$ and the HSP will be made.

\begin{subsection}{The Hidden Subgroup Problem for Abelian Groups}

We shall begin by recalling the statement of the Hidden Subgroup Problem in the finite Abelian group setting.

\begin{definition}{(The Hidden Subgroup Problem)}
Let $G$ be a finite Abelian group and $X$ a set. A function $f:G\to X$ is said to satisfy the hidden subgroup condition if there exists a subgroup $S\leq G$ such that
$$
f(g+s) = f(g)
$$
for all $g\in G$ and $s\in S$. In other words, $f$ can be defined coherently over $G/S$.

We call the Hidden Subgroup Problem for the function $f$ to that of obtaining $S$ using $f$ as a black box.
\end{definition}

\begin{remark}
Due to the nature of the $\GPK$, we need to ask for a further condition on $X$: we need it to be a finite Abelian group in order to apply the algorithm. 

However, this is no hindrance at all: as the condition and the problem themselves are independent of the group structure in the image set, we can simply define a finite Abelian group structure on $X$ (say, make it into a cyclic group of $\#(X)$ elements) and merrily go along.

Anyway, $X$ will indeed be a group for most of the applications of the HSP, the most famous one being the solution to the discrete logarithm problem using Shor's algorithm.
\end{remark}

Regarding the classical complexity of the HSP, it is well known that there are many instances whose associated decision problem is in the complexity class $\w{NP}$, such as the discrete logarithm problem. For more information on the classical complexity of this problem, \cite{hsp2} can be consulted.

\

The current strategy for solving the HSP for finite Abelian groups is a generalisation of Simon's algorithm, so we will take a small detour to quickly summarise the inner workings of this algorithm, as later we will compare it to the $\GPK$ based one.

Let $f:G\to H$ be the target function, we begin with two registers:
$$
\ket{\psi_1} = \ket{\w{0}}_G\otimes\ket{\w{0}}_H,
$$
so we can apply the Fourier transform gate, $\w{F}_G$, to the first register:
$$
\ket{\psi_2} = \left(\w{F}_G\otimes\w{I}_H\right)\ket{\psi_1} = \left(\frac{1}{\sqrt{|G|}}\sum_{g\in G}\ket{\w{g}}_G\right)\otimes\ket{\w{0}}_H.
$$

We continue by applying the oracle gate for $f$, $\w{U}_f$:
$$
\ket{\psi_3} = \w{U}_f\ket{\psi_2} = \frac{1}{\sqrt{|G|}}\sum_{g\in G} \ket{\w{g}}_G\otimes\ket{f(\w{g})}_H,
$$
and knowing that $f$ satisfies the HSP condition for some subgroup $S$ of $G$, we can rewrite this expression as
$$
\ket{\psi_3} = \frac{1}{\sqrt{|G|}}\sum_{g\in R}\sum_{s\in S} \ket{\w{g}+\w{s}}_G\otimes\ket{f(\w{g})}_H,
$$
where $R$ is a set of representatives for $G/S$. We finish by applying $\w{F}_G$ to the first register and measuring the second one. Let $g$ be such that we obtain $f(g)$ from the measurement, then the resulting superposition would be the following:
$$
\ket{\psi_4}_G = \w{F}_G\left(\frac{1}{\sqrt{|G|}}\sum_{s\in S} \ket{\w{g}+\w{s}}_G\right) = \frac{1}{|G|}\sum_{s\in S}\sum_{z\in G} \chi_{g+s}(z)\ket{\w{z}}_G
$$
$$
= \frac{1}{|G|}\sum_{z\in S}\chi_{g}(z)\left(\sum_{s\in G} \chi_s(z)\right)\ket{\w{z}}_G.
$$

Using now the fact that
$$
\sum_{s\in G} \chi_s(z) = \begin{cases} |S| & \text{if } \chi_z\in S^{\perp}\\
0 & \text{if } \chi_z\not\in S^{\perp}\end{cases},
$$
it is clear that the $\ket{\psi_4}$ is a superposition of states, $\ket{\w{z}}_G$ that satisfy $\chi_z\in S^{\perp}$ all of them with the same probability $|S|/|G|$ of being obtained after measuring the first register. The idea know is to repeat this process enough times until we get enough such $z$'s so that we can calculate $S$ by solving a linear system of equations.

\end{subsection}

\begin{subsection}{A $\GPK$ based algorithm}

Before stating the marker selection strategy that we will follow in the calls that we will make to the $\GPK$, let us take a moment to see what happens when we apply the $\GPK$ to a function $f:G\to H$ that satisfies the hidden subgroup condition.

\begin{lemma}
Let $G$ and $H$ be finite Abelian groups and let $f:G\to H$ be a function that satisfies the hidden subgroup condition. The result of applying the $\GPK$ to $f$ is a superposition of states $\ket{\w{z}}_G$ such that $\chi_z\in S^\perp$.
\end{lemma}

\begin{proof}
We know that, in general, the amplitude of the state $\ket{\w{z}}$ after applying the $\GPK$ to $f$ using $\chi_h\in\widehat{H}$ as a marker is
$$
\alpha_z = \frac{1}{|G|}\displaystyle\sum_{g\in G} \overline{\chi_h}\big(f(g)\big)\chi_g(z).
$$
Let now $R$ be a set of representatives of $G/S$, and using that $f(g) = f(g+s)$ for any $s\in S$, we can rewrite the previous expression as
$$
\alpha_z = \frac{1}{|G|}\displaystyle\sum_{g\in R}\sum_{s\in S} \overline{\chi_h}\big(f(g)\big)\chi_{g+s}(z) =\frac{1}{|G|}\displaystyle\sum_{g\in R} \overline{\chi_h}\big(f(g)\big)\sum_{s\in S} \chi_{g+s}(z),
$$
and using character properties we can see that
$$
\sum_{s\in S} \chi_{g+s}(z) = \sum_{s\in S} \chi_{g}(z)\chi_s(z) = \chi_g(z)\sum_{s\in S}\chi_s(z) = \chi_g(z)\sum_{s\in S}\chi_z(s)
$$
As $S$ is a subgroup, we know that
$$
\sum_{s\in S}\chi_z(s) = \begin{cases}
|S| & \text{if } \chi_z\in S^\perp\\
0 & \text{if } \chi_z\not\in S^\perp,
\end{cases}
$$
which implies that $\alpha_z = 0$ if $\chi_z\not\in S^\perp$.
\end{proof}

So, summing up, we get a superposition of the same states as in Simon's algorithm, but now the probabilities associated to each of the states are not necessarily the same, as it depends on the chosen marker $\chi_h$. 

The key now is to construct a selection strategy for the markers that maximises the probability of obtaining independent elements of $S^\perp$. The strategy that we will present is neither the most elegant nor the optimal one, but we will show that it allows for better probabilities of obtaining a new relevant element of $S^\perp$: we will simply choose a nontrivial marker completely at random, thus eliminating the possibility of $h$ being $0$ and reducing the total probability of obtaining $\chi_0$ after each round of the $\GPK$.

\vspace{4mm}

\textbf{$\GPK$ Algorithm for the HSP:}

\vspace{2mm}

$\mathbb{STEP}$ $1$:

\vspace{2mm}

Choose $\chi_h\in\widehat{H}\setminus\{\chi_0\}$ at random.

\vspace{2mm}

$\mathbb{STEP}$ $2$:

\vspace{2mm}

Apply $\GPK(f,\chi_h)$ and save the result, $\w{z}$, as $\chi_z\in S^\perp$.

\vspace{2mm}

$\mathbb{STEP}$ $3$:

\vspace{2mm}

Calculate $S$ from $S^\perp$ when enough elements in $S^\perp$ are gathered.

\end{subsection}

\begin{subsection}{Comparison}

We will now proceed to analyse and compare the performance of these two algorithms. In particular, we will focus on the probability of obtaining a new relevant element of $S^\perp$ in each of the iterations of the algorithms.

We have already established that in the case of Simon's algorithm for finite Abelian groups the probability of obtaining each element in $S^\perp$ is the same and equal to $|S|/|G|$. What happens then in the $\GPK$ inspired algorithm?

Let us first assume that, in Step 1, the marker is chosen completely at random in $\widehat{H}$ instead of in $\widehat{H}\setminus\{\chi_0\}$.

\begin{lemma}\label{equal}
Let $G$ and $H$ be Abelian groups and let $f:G\to H$ be an HSP function with hidden subgroup $S\subset G$. Then, the probability of obtaining $\chi_z\in\widehat{G}$ after choosing a marker $\chi_h\in\widehat{H}$ at random and performing $\GPK(f,\chi_h)$ is:
$$
p(z) = \begin{cases}
    |S|/|G| & \text{if } \chi_z\in S^\perp\\
    0 & \text{if } \chi_z\not\in S^\perp.
\end{cases}
$$
\end{lemma}

\begin{proof}
We know that the probability of obtaining $\chi_z\in\widehat{G}$ after an iteration of the $\GPK$ algorithm is
$$
p(z) = \frac{1}{|H|}\sum_{\chi_h\in \widehat{H}} \left|\alpha_{h}(z)\right|^2,
$$
where $\alpha_h(z)$ is the amplitude of $\ket{\w{z}}_G$ in the final superposition of the $\GPK$ using $\chi_h$ as a marker. We know that 
$$
\alpha_{h}(z) = \frac{1}{|G|}\displaystyle\sum_{g\in G} \overline{\chi_h}\big(f(g)\big)\chi_g(z),
$$
so we get that
$$
\begin{array}{rcl}
p(z) &=& \displaystyle \frac{1}{|H||G|^2} \sum_{\chi_h\in \widehat{H}}\left|\sum_{g\in G} \overline{\chi_h}\big(f(g)\big)\chi_g(z)\right|^2
\\ \\
&=& \displaystyle \frac{1}{|H||G|^2} \sum_{\chi_h\in \widehat{H}}\left(\sum_{g\in G} \overline{\chi_h}\big(f(g)\big)\chi_g(z)\right)\left(\sum_{g'\in G} \chi_h\big(f(g')\big)\overline{\chi_{g'}}(z)\right)
\\ \\
&=& \displaystyle \frac{1}{|H||G|^2} \sum_{\chi_h\in\widehat{H}} \sum_{g,g'\in G} \overline{\chi_h}\big(f(g)-f(g')\big)\chi_{g-g'}(z)
\\ \\
&=& \displaystyle \frac{1}{|H||G|^2}\sum_{g,g'\in G} \chi_{g-g'}(z) \sum_{\chi_h\in\widehat{H}}\overline{\chi_h}\big(f(g)-f(g')\big)
\end{array}
$$
If we focus just on the last sum,
$$
\sum_{\chi_h\in\widehat{H}}\overline{\chi_h}\big(f(g)-f(g')\big) = \begin{cases}
    0 & \text{if } f(g)\neq f(g')\\
    |H| & \text{if } f(g) = f(g'),
\end{cases}
$$
and as $f$ is a function satisfying the HSP condition, we get that
$$
p(z) = \frac{1}{|G|^2}\sum_{g\in G}\sum_{s\in S}\chi_s(z) = \frac{1}{|G|} \sum_{s\in S}\chi_s(z) = \begin{cases}
    |S|/|G| & \text{if } \chi_z\in S^\perp\\
    0 & \text{if } \chi_z\not\in S^\perp.
\end{cases}
$$
\end{proof}

However, if the chosen marker is $\chi_0\in\widehat{H}$, it is clear that the final superposition in the $\GPK$ is just the state $\ket{\w{0}}_G$, as
$$
\alpha_{\w{0}} =  \frac{1}{|G|}\displaystyle\sum_{g\in G} \overline{\chi_0}\big(f(g)\big) = \frac{|G|}{|G|} = 1.
$$

This implies that, if we were to use $\chi_0$ as a marker, then we would get the trivial character in $S^\perp$, which does not provide any valuable information. By eliminating the possibility of choosing $\chi_0$ as a marker we are reducing the probability of obtaining $\chi_0$ as a result after the iteration, which in turn improves the probability of obtaining a relevant character after each iteration.

\begin{theorem}\label{unequal}
Let $f:G\to H$ be finite Abelian group function satisfying the HSP condition with secret subgroup $S\subset G$. The probability of obtaining $\chi_z\in\widehat{G}$ as a result in each iteration of the $\GPK$ for the HSP using a random marker $\chi_h\in\widehat{H}\setminus\{\chi_0\}$ is:
$$
p(\w{z}) = \begin{cases}
    \displaystyle \frac{|S||H|-|G|}{|G|(|H|-1)} & \text{if } z = 0\\[12pt]
    \displaystyle \frac{|S||H|}{|G|(|H|-1)} & \text{if } \chi_z\in S^\perp\\[12pt]
    0 & \text{if } \chi_z\not\in S^\perp.
\end{cases}
$$
\end{theorem}

\begin{proof}
Let $p_h(z)$ is the probability of getting $\chi_z\in\widehat{G}$ as a result after applying the $\GPK$ using $\chi_h\in\widehat{H}$ as a marker. Following the notation and calculations of Lemma \ref{equal}, we showed that
$$
\frac{1}{|H|}\sum_{\chi_h\in\widehat{H}} p_h(z) = \frac{|S|}{|{G|}}
$$
if $\chi_z\in S^\perp$. Using now the fact that $p_0(0) = 1$, we can see that
$$
\sum_{\chi_h\in\widehat{H}\setminus\{0\}} p_h(0) = \frac{|S||H|}{|G|}-1 = \frac{|S||H|-|G|}{|G|},
$$
which implies that the probability of obtaining $\chi_0$ when the random selection is among non-trivial characters is
$$
\frac{1}{|H|-1}\sum_{\chi_h\in\widehat{H}\setminus\{0\}} p_h(0) = \frac{|S||H|-|G|}{|G|(|H|-1)}.
$$
Analogously, for any $z$ such that $\chi_z\in S^\perp$, we know that $p_0(z) = 0$, so the probability of obtaining any such $\chi_z$ after an iteration of the $\GPK$ using random selection among non-trivial characters is
$$
\frac{1}{|H|-1}\sum_{\chi_h\in\widehat{H}\setminus\{0\}} p_h(z) = \frac{1}{|H|-1}\sum_{\chi_h\in\widehat{H}} p_h(z) = \frac{|S||H|}{|G|(|H|-1)}.
$$
\end{proof}

Thus, we are essentially busting up the probability of obtaining non-trivial characters by a factor of $|H|/(|H|-1)$. Furthermore, the random selection strategy that we have analysed here is rather naive, so a more elegant strategy might offer even more improvement.

\end{subsection}

\end{section}

\begin{section}{FBI Functions for Finite Abelian Groups}\label{cuatro}

In the introduction, it was mentioned that one of the more relevant applications of the $\GPK$ algorithm arose when it was fed with functions from a relevant class that was defined and studied in \cite{waifi}: fully balanced image functions or FBI functions.

In this section, we will generalise the concept of FBI functions to the Abelian groups situation, proving a relevant characterization which will then be used to solve the FBI problem of determining the dimension of the function image using the $\GPK$ algorithm.

\begin{subsection}{Fully Balanced Image Functions}

Before all that, we will need to make some definitions that will help us in extending the concept of FBI functions to Abelian groups.

\begin{definition}{($\chi$-constant sets)}
    Let $G$ be an Abelian group $G$, $\chi\in \widehat{G}$ and let $S$ be a nonempty subset of $G$. We will say that $S$ is $\chi$-constant if $\chi|_S$ is constant.
\end{definition}
    
\begin{remark}
    It is easy to see that $S$ is $\chi$-constant if and only if
    $$
    \left|\widehat{1_S}(\chi)\right| = \left| S \right|. 
    $$

    On the one hand, being $S$ $\chi$-constant, we have $\chi(g) = \gamma \in \mu_n$ for all $g \in S$, therefore
    $$
    \left|\widehat{1_S}(\chi)\right| = \left|\sum_{g\in S} \chi(g)\right| = \left| S \right| \cdot \left| \gamma \right| = \left| S \right|.
    $$

    On the other hand, if $S$ is not $\chi$-constant the triangle inequality becomes strict and we must have
    $$
    \left|\widehat{1_S}(\chi)\right| = \left|\sum_{g\in S} \chi(g)\right| < \sum_{g\in S} \left| \chi(g)\right| = \left| S \right|.
    $$
\end{remark}

The opposite behavior to $\chi$-constant sets is that of $\chi$-balanced, although the definition parallels the characterization above.
    
\begin{definition}{($\chi$-balanced and fully balanced sets)}
    Let $G$ be an Abelian group $G$, $\chi\in \widehat{G}$ and let $S$ be a nonempty subset of $G$. We will say that $S$ is $\chi$-balanced if 
    $$
    \widehat{1_S}(\chi) = 0.
    $$

    We will also say that $S$ is fully balanced if it is either $\chi$-constant or $\chi$-balanced for every $\chi\in\widehat{G}$.   
\end{definition}


\begin{remark}\label{G-bal}    
    From the definition, $S$ being balanced is somehow granted if $\chi(S)$ is equidistributed along $\chi(G) = \mu_\ell = \{ \gamma_1,\ldots,\gamma_\ell\}$. In particular, as we saw previously, $G$ is always $\chi$-balanced for any character $\chi$, except for the trivial one $\chi_0 = 1_G$, for which it is $\chi_0$--constant. This is entirely predictable as
    $$
    \left| \chi^{-1} (\gamma_i) \right| = |\ker(\chi)|, \qquad i = 1,\ldots,\ell.
    $$
    Analogously, all subgroups $K \leq G$ are $\chi$--balanced, as long as $\chi|_K \neq 1_K$. That is, all subgroups are fully balanced sets.

    \ 
    
    However, equidistribution is not a necessary condition. Let us consider
    $$
    \chi: \mathbb{Z} / 16 \mathbb{Z} \to S^1 \subset \mathbb{C}, \qquad \chi(x) = i^x,
    $$
    which is obviously a character. Then, the set
    $$
    S = \{ 0,\ 1,\ 2,\ 3,\ 5,\ 7,\ 13,\ 15 \} \subset \mathbb{Z} / 16 \mathbb{Z}
    $$
    is $\chi$-balanced, as
    $$
    \widehat{1_S}(\chi) = \sum_{x\in S} \chi(x) = 0,
    $$
    but the images are not equidistributed, as $\pm 1$ appear once, while $\pm i$ appear three times apiece.

    The full geometric interpretation of $\chi$-balanced sets comes from the Rédei--de Brujin--Schoenberg Theorem \cite{redei,brujin,schoenberg}. We have $S$ is $\chi$--balanced if and only if
    $$
    \sum_{x \in S} \chi(x) = 0,
    $$
    so we are left with a sum--zero subset of $\mu_\ell$. 
    
    Then said theorem shows that any such set can be written as set unions and/or set differences of full sets of roots $\mu_t$ (with $t|\ell$), eventually rotated. In our previous example we had a full set $\mu_4$ (comprising $\{\pm1, \pm i \}$) and two extra sets $\{ \pm i\}$ which is the set $\mu_2 = \{ \pm 1\}$ after a $\pi/2$--rotation. Of course one can also see it as three full sets $\mu_4$ minus two sets $\mu_2$ (without any rotations involved).

    \ 

    Using this characterization is easy to see, for example, that
    $$
    S = \{ 1, \ 3 \} \subset \mathbb{Z} / 4 \mathbb{Z} 
    $$
    is a fully balanced set which is not a subgroup.
\end{remark}

These notions can be easily extended to multisets, changing the group $G$ for a multiset $(G,m)$. In particular, given $\chi \in \widehat{G}$ we will say that:
\begin{itemize}
    \item $(G,m)$ is $\chi$-constant if $\chi|_{\sup(G,m)}$ is constant. This is equivalent to
    $$
    \left|\widehat{m}(\chi)\right| = \left| \sum_{g\in G} m(g) \chi(g) \right| = \sum_{g \in \sup(G,m)} m(g).
    $$
    \item $(G,m)$ is $\chi$-balanced if
    $$
    \widehat{m}(\chi) = \sum_{g\in G} m(g) \chi(g) = 0.
    $$
    \item $(G,m)$ is fully balanced if it is either $\chi$-constant or $\chi$-balanced for every $\chi\in\widehat{G}$.
\end{itemize}

\begin{remark}
Note that the presence of a multiplicity function prevents $(G,m)$ to be automatically $\chi$-balanced.

On the other hand, as we allow multiplicities to be $0$, this definition actually deals also with subsets: if we are interested in studying these properties on a particular subset $S \subset G$, we can simply change $m|_S$ to be $1$, $m|_{G\setminus S}$ to be $0$ and consider the new multiset.
\end{remark}

This multiset framework is particularly interesting for us as, associated to a function $f:G \to H$ between Abelian groups we can define a multiset $(H,m)$ with the multiplicity given by
$$
m(h) = | f^{-1} (h) |, \qquad \forall h \in H,
$$
where obviously $\sup(H,m) = \img(f)$. We can now establish our main concept.

\begin{definition}{(FBI function)}
    Let $G$ and $H$ be Abelian groups, and $f:G\to H$ a function. Then $f$ is said to be a fully balanced image function (or simply is said to be FBI) if the multiset $(H,m)$ associated to $f$ is fully balanced.
\end{definition}

In other words, a function $f:G\to H$ is an FBI function if
$$
\left|\sum_{g\in G} \chi(f(g))\right| \in \big\{ 0, \; \left|G\right| \big\}, \qquad \forall \chi\in\widehat{H}.
$$

\begin{remark}
In particular, group homomorphisms are always FBI functions, as in this case the multiplicity function is given by
$$
m(h) = \left\{ \begin{array}{ll} 
|\ker(f)| & \text{ if } \; h \in \img(f) \\
0 & \text{ otherwise}
\end{array} \right.
$$
so one can prove $(H,m)$ is fully balanced in the same way subgroups are. 

\ 

However there are more examples of FBI functions. For instance, we can consider a group homomorphism $f:G\to H$, preceded by any permutation in $G$ which only swaps elements with the same image by $f$. These examples actually set a template for a characterization of FBI functions.
\end{remark}

As it is customary in group theory, given a subgroup $K \leq H$ (still in additive notation) we will call an orbit of $K$ (or a translation of $K$) any set of the form
$$
h+K = \big\{ h+k \; | \; k \in K \big\},
$$
for some $h \in H$.

For technical reasons we will notate the set of $\chi\in\widehat{H}$ such that $f$ is $\chi$-constant as $C(f)$, and call it the \textit{constant set}, and the set of all those $\chi\in\widehat{H}$ that balance $f$ as $B(f)$, and call it the \textit{balancing set}. Therefore $f:G\to H$ is FBI if and only if
$$
C(f)\cup B(f) = \widehat{H}.
$$

It is easy to show that $C(f)$ is a subgroup of $\widehat{H}$, as $\chi_0$ is in $C(f)$ and if $\chi_1|_{\img(f)}$ and $\chi_2|_{\img(f)}$ are constant, then $(\chi_1\chi_2)|_{\img(f)}$ is also constant. Furthermore, $B(f)$ is a disjoint union of orbits of $C(f)$, as if $\chi_1\in C(f)$ then it is constant in $\img(f)$ with value $c$ and if $\chi_2\in B(f)$, then
$$
\sum_{g\in G} \left(\chi_1\chi_2\right)\big(f\left(g\right)\big) = \sum_{g\in G} \chi_1\big(f\left(g\right)\big)\chi_2\big(f(g)\big) = c\sum_{g\in G} \chi_2\big(f\left(g\right)\big) = 0,
$$
so $\chi_1\chi_2$ is in $B(f)$.

\begin{theorem}{(Characterization of FBI functions)}\label{FBIchar}
    Let $G$ and $H$ be finite Abelian groups and $f:G\to H$ be a function. Then, $f$ is FBI if and only if $\img(f)$ is the orbit of a subgroup of $H$ and $\left|f^{-1}(y)\right|$ is constant for $y\in\img(f)$.
\end{theorem}

\begin{proof}

Let us consider the multiset corresponding to the image of $f$, $(H,m)$, as described before. Then, we want to prove that $(H,m)$ is fully balanced if and only if $\sup(H,m) = \{x\in H \mid m(x) \neq 0\}$ is the orbit of a subgroup in $H$ and $m$ is constant in $\sup(H,m)$.

The converse implication is easier to prove as, under the hypothesis, one has
$$
\sum_{h\in H} m(h)\chi(h) = c \sum_{h\in\sup(m)} \chi(h) =\begin{cases} 0 & \text{ if } \; \chi \neq 1_{H},  \\ |H| & \text{ if } \; \chi = 1_H, \end{cases}
$$
where $c$ is the constant value for $m$ in $\sup(m)$ .

For the direct implication, we can proceed as follows. Let us write
$$
c_m = \sum_{g\in G} m(g).
$$

As $(H,m)$ is fully balanced, we know that 
$$
\widehat{m}(\chi) = \sum_{g\in G} m(g) \chi(g)  \in \big\{ 0,\ -c_m, \ c_m \big\}
$$ 
for any $\chi\in\widehat{H}$. We will also assume--without loss of generality--that $m(0) \neq 0$. As $m$ is not the zero function and both properties (being fully balanced and being an orbit) are preserved by translation, we are simply forcing our orbit to be \textit{the} subgroup. In that situation, $-c_m$ is not a possible value for $\widehat{m}(\chi)$, because $0$ has non-zero multiplicity and $\chi(0) = 1$.

It is clear that, by definition,
$$
\widehat{m} = c_m 1_{C(f)}.
$$

Using now the inverse discrete Fourier transform formula, we get
$$
\widehat{\widehat{m}} = |G|m = |c_m| \widehat{1}_{C(f)}.
$$
As $C(f)$ is a subgroup, we have that
$$
|c_m|\widehat{1_{C(f)}} = |c_m| \cdot |C(f)| \cdot 1_{C(f)^{\perp}},
$$
so $\sup(m) = \img(f) = C(f)^{\perp}$ is a subgroup of $G$ and $m$ is constant in $C(f)^{\perp}$ and zero elsewhere.
\end{proof}

It should be remarked that, if $f$ is FBI, then $C(f)^{\perp}$ is the underlying subgroup of $\img(f)$, a fact that will be used in the following subsection.

\begin{corollary}
    Let $(G,m)$ be a multiset, where $G$ is an Abelian finite group. $(G,m)$ is fully balanced if and only if $S = \sup(G,m)$ is an orbit of a subgroup of $G$ and $m|_S$ is constant.
\end{corollary}

\end{subsection}

\begin{subsection}{The FBI Problem}

Let us focus now on the problem that will be solved using the generalised version of the $\GPK$ algorithm for finite Abelian groups. 

\begin{definition}{(FBI problem for finite Abelian groups)}
Let $G$ and $H$ be finite Abelian groups and $f:G\to H$ be an FBI function. The FBI problem for $f$ will be that of determining the order of the underlying group of $\img(f)$ using $f$ as a black box.
\end{definition}

An immediate remark should be made on the classical complexity of this task.

\begin{remark}
The classical complexity of this problem highly varies with the choice of groups $G$ and $H$. As an example, taking into account that $|\img(f)|$ must divide both $|G|$ and $|H|$. So, if $\mathrm{mcd}(|G|,|H|) = 1$, then the subgroup $\img(f)$ can only be the trivial one, so it is actually known without making any calls to $f$. 

However, if $G$ and $H$ happen to be the same group, then the number of calls needed in the worst case scenario is $\mathcal{O}(|G|/p)$, where $p$ is the order of the smallest subgroup in $\img(f)$. Another relevant situation is when $G = \left(\mathbb{Z}/p\mathbb{Z}\right)^n$ and $H = \left(\mathbb{Z}/p\mathbb{Z}\right)^m$ for a prime $p$. Here, the number of calls needed is always $\mathcal{O}(p^{n-1})$.
\end{remark}

To solve this problem using a quantum computer we will mainly make use of the following result.

\begin{lemma}
Applying the finite Abelian group version of the $\GPK$ algorithm to an FBI function, $f:G\to H$, with marker $\chi_h$ with $h\in H$ yields $0$ as a result if $f$ is $\chi_h$-constant and any other result if $f$ is $\chi_h$-balanced.
\end{lemma}

\begin{proof}
We know that the amplitude of the state $\ket{\textbf{z}}_H$ in the final superposition of the $\GPK$ algorithm applied to $f$ with marker $\chi_h$ is:
$$
\alpha_z = \frac{1}{|G|}\displaystyle\sum_{g\in G} \overline{\chi_h}\big(f(g)\big)\chi_g(z),
$$
so in particular, the amplitude of $\ket{0}_G$ is:
$$
\alpha_0 = \frac{1}{|G|}\displaystyle\sum_{g\in G} \overline{\chi_h}\big(f(g)\big).
$$
It follows trivially that if $f$ is $\chi_h$-constant then $\alpha_0 = 1$ and that if it is $\chi_h$-balanced, then $\alpha_0 = 0$.
\end{proof}

The overview of the algorithm that we will present shortly is as follows. Several applications of the $\GPK$ will be performed using the target function $f$ but changing the markers $\chi_{h_i}\in\widehat{H}$. Depending on the result of the algorithm, each $\chi_{h_i}$ will be stored as part of the constant set or the balancing set of $f$. After some rounds of this process, we will have gathered enough information to solve the FBI problem.

Taking all of this into account, the only thing left to do create a marker selection strategy that minimises the number of calls that will have to be made to the $\GPK$. The following is a description of said strategy.

\vspace{4mm}

\textbf{Marker Selection Algorithm:}

\vspace{2mm}

$\mathbb{STEP}$ $1$:

\vspace{2mm}

$C = \varnothing$ and $B = \varnothing.$

We begin by setting our two sets to be empty.

\vspace{2mm}

$\mathbb{STEP}$ $2$:

\vspace{2mm}

Take a non-trivial $\chi_h\in\widehat{H}$ such that $\chi_h\not\in\langle C\rangle$, $\chi_h\not\in bC$ for any $b\in B$ and such that $B\cap\langle\chi_h\rangle = \varnothing$ and compute $\GPK(f,\chi_h)$.

If the result is $\mathbf{0}$, then add $\chi_h$ to $C$ and repeat Step $2$.

Else, run Step $3$.

\vspace{2mm}

$\mathbb{STEP}$ $3$:

\vspace{2mm}

If $B = \varnothing$, then add $\chi_h$ to $B$ and go to Step $5$.

Else, for each $\chi_y\in B$, consider $\chi_{h'} = \overline{\chi_y}\chi_h$ and run Step $4$.

\vspace{2mm}

$\mathbb{STEP}$ $4$:

\vspace{2mm}

Compute $\GPK(f,\chi_{h'})$. If the result is $\mathbf{0}$, then add $\chi_{h'}$ to $C$ and go to Step $5$. Else, continue with the next element of $B$ in Step $3$.

If for every $\chi_y\in B$ we get a result different from $\mathbf{0}$, add $\chi_{h}$ and every $\chi_{h'}$ to $B$ and go to Step $5$.

\vspace{2mm}

$\mathbb{STEP}$ $5$:

\vspace{2mm}

Repeat Step $2$ until $|\langle C\rangle|\cdot(|B|+1) = |H|$. The order we are looking for is 
$$
|\img(f)| = \frac{|H|}{|\langle C\rangle|} = |B| + 1.
$$

Note that $C$ and $B$ are not storing the whole constant and balancing sets, but rather a generating set and a set of orbit representatives respectively.

\begin{remark}
Let us take a look at an example of how the algorithm works. Consider the function $f:\mathbb{Z}/12\mathbb{Z}\to\mathbb{Z}/12\mathbb{Z}$ given by:
$$
\begin{array}{cccc}
    f(0) = 0 & f(1) = 3 & f(2) = 3 & f(3) = 9 \\
    f(4) = 9 & f(5) = 3 & f(6) = 0 & f(7) = 6  \\
    f(8) = 0 & f(9) = 6 & f(10) = 6 & f(11) = 9
\end{array}
$$
It is clearly an FBI function because the image is a subgroup ($\img(f) = \{0,3,6,9\}$) and the multiplicities are all the same. It should also be noted that, although here we are showing the function $f$ explicitly for explanatory purposes, we have to imagine it as just a black box. 

For this example, the character associated with the element $k\in\mathbb{Z}/12\mathbb{Z}$ is chosen to be the one defined as:
$$
\chi_k(j) = \exp\left(\frac{2\pi i\,k\cdot j}{12}\right).
$$

Following the instructions of the marker selection algorithm, we take any non-trivial element of $\widehat{H}$, for instance, $\chi_1$. After applying the $\GPK$ using said marker, we get that $\chi_1$ balances $f$, so we should add it to $B$. Note that this piece of information discards the possibility that $f$ is constant, and also that any other $\chi_k$ such that $1$ is in $\langle k\rangle$ must also balance $f$, although we cannot add it to $B$ because we do not know whether they represent a new orbit or not.

On that note, we should now take another character $\chi_k$ such that $k|12$, so let us consider, for example, $\chi_2$. It is clear that $\chi_2$ also balances $f$, so after we apply the $\GPK$ on more time we should get a result different form $\w{0}$, but before adding it to $B$ we must verify that the do not represent the same orbit, i.e, that $\chi_1\overline{\chi_2} = \chi_{11}$ is not in $C(f)$. In this case, it is trivial that this is not the case, as $1\in\langle 11\rangle$, but in general we should have needed to make an extra call to the $\GPK$. In any case, we can add $\chi_2$ and $\chi_{11}$ to $B$.

Let us now choose a new character following the in instructions of the algorithm. For instance, let us take $\chi_4$, which should return $\w{0}$ when being used as an argument to the $\GPK$, so we can add it to $C$. The question now is whether we are finished, but we can easily check that we are, as $|\langle\chi_4\rangle| = 3$ and $|B|+1 = 4$. This also implies that $\img(f)$ has order 
$$
\frac{|H|}{|\langle C\rangle|} = \frac{12}{3} = 4.
$$
Furthermore, we know that the underlying subgroup of $\img(f)$ is  $C(f)^\perp = \langle C\rangle^\perp$, so we can calculate a complete description of the image of the function.

Comparing the performance of our algorithm to the classical one, it can be seen that we would have needed six calls to $f$ in the worst case scenario to achieve our goal, while we have done it in just three (although we would have done it in four if we had considered $\chi_3$ instead of $\chi_4$).
\end{remark}

Let us check that the algorithm correctly solves the problem and analyse its worst case oracle complexity.

\begin{theorem}{(Correctness of the marker selection algorithm)}
The previous algorithm correctly solves the FBI problem for a function $f$ in at most $$
\big(|\img(f)|-1\big)\left(\log_2\left(\frac{|H|}{|\img(f)|}\right)+1\right)
$$
calls to the $\GPK$.

\end{theorem}

\begin{proof}
Firstly, we are adding to $C$ only independent elements of $C(f)$, as Step 2 guarantees independence and we only add $\chi$ to $C$ if $\GPK(f,\chi)$ returns $\w{0}$.

Secondly, we are only adding representatives of different orbits of $B(f)$ to $B$, as we only add characters $\chi$ such that $\GPK(f,\chi)$ does not return $\w{0}$ and Step 4 guarantees that the character is not part of a previously considered orbit in $B$. Considering we stop when $(|B|+1)\langle C\rangle = |H|$, we know that we will get a generating set for $C(f)$ and a set of representatives for $B(f)$.

Lastly, the worst case scenario is that we obtain first all the elements in $B$ and then we obtain all the elements in $C$ in the last check of Step 4. This implies making
$$
|B| = |\img(f)|-1
$$
calls to the $\GPK$ for $B$ in the worst case scenario and afterwards making an extra of
$$
|B|\log_2\left(\frac{|H|}{|\img(f)|}\right)
$$
calls to the $\GPK$, as we will need at most 
$$
\log_2\big(|C(f)|\big) = \log_2\left(\frac{|H|}{|\img(f)|}\right)
$$
elements to generate $C(f)$. Adding these two expressions together we get
$$
\big(|\img(f)|-1\big)\left(\log_2\left(\frac{|H|}{|\img(f)|}\right)+1\right).
$$
\end{proof}

Of course, this result implies that we have a nice improvement on the deterministic oracle complexity in the quantum situation over the classical case when $\img(f)$ is of small size when compared to $G$.

Furthermore, this algorithm and marker selection strategy show a path to solve similar problems for different classes of finite Abelian group functions $f:G\to H$ that have different Fourier transform supports.

\end{subsection}

\end{section}

\begin{section}{Conclusion and Further Research}

The main contributions that we have presented here are the generalisation of the $\GPK$ algorithm to the setting of finite Abelian groups, the solution using this new algorithm to solve the hidden subgroup problem for finite Abelian groups with better probability than using Simon's algorithm and the statement and solution of the fully-balanced image problem.

\

Our hope is to further generalise these ideas to the case of finite non-Abelian groups, as the hidden subgroup problem for non-Abelian groups has not being solved in general \cite{hsp,calixto}. There are two main complications in achieving this goal: on the one side, representations of non-Abelian groups are not one-dimensional---although the Fourier transform can be defined without many complications, as shown in \cite{calixto}---. On the other side, one-dimensional characters do not distinguish general non-Abelian subgroups. 

\

However, there are also open questions regarding these two final contributions. For the HSP algorithm we have employed a very naive strategy of randomly choosing a non-trivial character, so perhaps a more elegant strategy that makes use of the structures of $S^\perp$ and $\widehat{H}$ might further improve our solution. For the FBI problem, the selection strategy that we propose is optimal among marker algorithms for some groups, but the general landscape is too vast to prove optimality without assuming some conditions on the finite Abelian groups. Furthermore, the idea of distinguishing classes of functions using the balancing and constant sets can be extended to distinguish any two classes of functions with disjoint Fourier support.

\end{section}

\section*{Acknowledgements}

This article is funded by IMUS - María de Maeztu grant CEX2024-001517-MApoyo a Unidades de Excelencia María de Maeztu, by grants SOL2024-31596 and SOL2024-31708 from Plan Propio de Investigación y Transferencia de la Universidad de Sevilla, cofunded by the EU - Ministerio de Hacienda y Función Pública - Fondos Europeos - Junta de Andalucía – Consejería de Universidad, Investigación e Innovación (second author) and by MICIU/AEI/10.13039/501100011033, grant PID2024-156912N funded by Ministerio de Ciencia e Innovación (both authors).

\section*{Competing interests}

The authors report that there are no competing interests to declare.

\bibliographystyle{nature}
\bibliography{GPK3}

\end{document}